\documentclass[12pt]{article}
\usepackage[small]{titlesec}
\usepackage{natbib,amsmath,amssymb,amsthm,graphicx,setspace,paralist,booktabs,rotating,subcaption,float,color}
\usepackage{graphicx}
\usepackage{enumerate}
\usepackage{url} 
\usepackage{amssymb, amsmath, amsthm}
\usepackage{breakcites}

\usepackage{bigstrut} 
\usepackage{blkarray}
\usepackage{booktabs}
\usepackage{enumitem} 
\usepackage{float} 
\usepackage{graphicx}
\usepackage{subcaption}
\usepackage{mwe}
\usepackage{comment}
\usepackage[colorlinks=true, allcolors=LinkBlue, backref=page]{hyperref}
\usepackage{lmodern}
\usepackage{mathrsfs} 
\usepackage{multirow}
\usepackage{natbib} 
\usepackage{scalerel}
\usepackage{setspace}
\usepackage{soul} 
\usepackage{subcaption}
\usepackage{titling}
\usepackage{verbatim}
\usepackage{xcolor}
\usepackage{xfrac}
\usepackage{algpseudocode}
\usepackage{bbm}
\usepackage{amsfonts}
\usepackage{outlines}

\AtBeginDocument{
	\abovedisplayskip=5pt plus 2pt minus 2pt
	\belowdisplayskip=\abovedisplayskip
	\abovedisplayshortskip=2pt plus 2pt minus 2pt
	\belowdisplayshortskip=\belowdisplayskip}

\titlespacing*{\section}{0pt}{*2}{*1}
\titlespacing*{\subsection}{0pt}{*2}{*1} 
\bibpunct{(}{)}{;}{a}{}{,}
\setlist{noitemsep, topsep=0pt} 
\definecolor{LinkBlue}{rgb}{.15, .25, .85} 

\binoppenalty=\maxdimen 
\relpenalty=\maxdimen 

\makeatletter
\renewcommand*{\NAT@spacechar}{~}
\makeatother

\usepackage{algorithm}
\makeatletter

\newfloat{algorithm}{tbp}{loa}
\providecommand{\algorithmname}{Algorithm}
\floatname{algorithm}{\protect\algorithmname}

\newtheorem{theorem}{Theorem}
\newtheorem{example}{Example}

\newtheorem{remark}{Remark}
\newtheorem{corollary}{Corollary}
\newtheorem{definition}{Definition}

\newtheorem{prop}{Proposition}

\def \bP {\mathbb{P}}
\def \bE {\mathbb{E}}

\usepackage{xspace}

\newcommand{\ceil}[1]{{\left\lceil {#1} \right \rceil}}

\newcommand{\TV}{{\sf TV}}

\newcommand{\diff}{\mathrm{d}}

\newcommand{\Indc}{\mathbf{1}}

\newcommand{\calA}{{\mathcal{A}}}

\newcommand{\calD}{{\mathcal{D}}}

\newcommand{\calO}{{\mathcal{O}}}

\newcommand{\calS}{{\mathcal{S}}}

\newcommand{\barP}{{\bar{P}}}

\newcommand{\rMLMC}{\mathsf{rMLMC}}
\newcommand{\uMCMC}{\mathsf{uMCMC}}

\newcommand{\MCMC}{\mathsf{MCMC}}

\begin{document}
	\title{When are Unbiased Monte Carlo Estimators More Preferable than Biased Ones?}
 \author{Guanyang Wang, Jose Blanchet,  Peter W.Glynn}
 \maketitle
\begin{abstract}
    Due to the potential benefits of parallelization, designing unbiased Monte Carlo estimators, primarily in the setting of randomized multilevel Monte Carlo, has recently become very popular in operations research and computational statistics. However, existing work primarily substantiates the benefits of unbiased estimators at an intuitive level or using empirical evaluations. The intuition being that unbiased estimators can be replicated in parallel enabling fast estimation in terms of wall-clock time. This intuition ignores that, typically, bias will be introduced due to impatience because most unbiased estimators necesitate random completion times. This paper provides a mathematical framework for comparing these methods under various metrics, such as completion time and overall computational cost. Under practical assumptions, our findings reveal that unbiased methods typically have superior completion times — the degree of superiority being quantifiable through the tail behavior of their running time distribution — but they may not automatically provide substantial savings in overall computational costs. We apply our findings to Markov Chain Monte Carlo and Multilevel Monte Carlo methods to identify the conditions and scenarios where unbiased methods have an advantage, thus assisting practitioners in making informed choices between unbiased and biased methods. 
\end{abstract}

\section{Introduction}

Due to the potential of parallelization, designing unbiased Monte Carlo estimators has recently received much attention. Among existing methods, two categories of algorithms stand out: the unbiased Markov chain Monte Carlo (MCMC) \citep{glynn2014exact,jacob2020unbiased,heng2019unbiased,middleton2020unbiased}, and unbiased Multilevel Monte Carlo (MLMC) \citep{rhee2015unbiased, blanchet2015unbiased, blanchet2019unbiased, vihola2018unbiased}.  The unbiased MCMC and MLMC are considered the unbiased counterparts to their biased MCMC \citep{brooks2011handbook} and MLMC \citep{heinrich2001multilevel, giles2015multilevel} equivalents. They have found widespread applications in operations research, statistics, and machine learning, including MCMC convergence diagnosis \citep{biswas2019estimating}, stochastic optimization \citep{blanchet2019unbiased,asi2021stochastic}, optimal stopping \citep{zhou2021unbiased} and so on.

While the advantages of these unbiased estimators for parallelization are intuitively clear, and recent studies such as \cite{nguyen2022many, wang2022unbiased} have presented empirical evidence supporting these benefits, there still lacks a systematic comparison between these unbiased Monte Carlo methods and their biased counterparts within a mathematical framework. Establishing such a framework involves two complexities. 

Firstly, we need to define a metric that is pragmatic for practitioners. For instance, in situations where users are acquiring parallel processors based on the on-demand pricing model from vendors like Amazon or Microsoft, a suitable metric might be the total monetary cost, which is proportional to the total computational cost across all processors. In many other cases, e.g., when the users have an urgent need for the computational result, the completion time would be a better metric. Secondly, almost all unbiased Monte Carlo methods exhibit a random yet finite computational cost per implementation, contrasting with traditional methods that usually have a deterministic cost. This inherent randomness can substantially impact our comparison. To illustrate, if we choose completion time as our metric, the completion time for a conventional method would be the same across all processors (in equal environments), but each unbiased estimator generally carries a random cost (i.e. a random termination time). Hence, the comparative result, utilizing the completion time of the slowest processor, may vary from that employing the average completion time across all processors.

The main contribution of this paper is a systematic comparison between the unbiased and biased Monte Carlo methods in the massively parallel regime under the metrics mentioned above, with a particular focus on the unbiased MCMC and unbiased MLMC methods since they have already been used in practice. After fixing an additive desired precision level $\epsilon$, we compare the behavior of these two types of algorithms based on their scaling to $\epsilon$. Under practical assumptions, unbiased algorithms generally have a favorable comparison to their biased counterparts concerning both worst and average completion time. The extent of this advantage is quantifiably linked to the tail behavior of the running time needed to implement the unbiased algorithm once. Conversely, unbiased estimators do not offer savings in total computational cost compared to their biased equivalents. However, both methodologies frequently attain the same order of magnitude on  $\epsilon$. Consequently, their differences are mostly evident in constants. In summary, our results not only validate the circumstances where these unbiased methods are advantageous for parallel implementation but also clarify the cases in which unbiased methods do not have such an advantage. We hope these results shed some light on which estimator to use in practice. 

To intuitively understand the benefits of unbiased Monte Carlo methods regarding completion time let us use the setting of MLMC as a comparison environment. When applying a standard (i.e. biased) MLMC strategy users set a precision level, say \( \epsilon \), and select the number of levels so that the combined squared bias and variance do not exceed \( \epsilon^2 \). In contrast, unbiased estimators employ a distinct strategy to control bias and variance independently. Typically, an efficient unbiased Monte Carlo algorithm requires a random time (with finite expectation) in order to produce an estimator with zero bias and finite variance. Users then manage the variance by independent replications of the unbiased estimators. Averaging the outcomes of \( N \) independent runs will result in a variance that is \( \frac{1}{N} \) of a single estimator's. Consequently, the MSE of the concluding estimator will linearly decrease with the number of repetitions. Given that unbiased Monte Carlo estimators cleanly separate the issue of bias from the issue of variance, the intuition indicates that they are inherently more suitable for parallel implementation compared to their biased counterparts. This intuition will be quantified in the following three sections. Our quantitative results are illustrated both in the context of MLMC and MCMC applications.

In the context of multilevel Monte Carlo applications, the main insights can be summarized as follows. We consider three estimators. First, the standard (biased) multilevel Monte Carlo estimator studied in \citep{giles2008multilevel, giles2015multilevel} and reviewed in Section \ref{sec:MLMC}. This estimator is not designed to be applied directly to a parallel computing environment. The second estimator is a natural adaptation of the standard estimator designed to minimize completion time so that the workload is distributed in such a way that every parallel processor is treated equally. The third estimator is the unbiased randomized multilevel Monte Carlo estimator (rMLMC) introduced and studied in \citep{rhee2015unbiased}. It is well known (also discussed in Section \ref{sec:MLMC}) that in terms of expected computational cost (when the cost achieves the canonical $1/\epsilon^2$ rate, see Theorem \ref{thm:MLMC}), the standard estimator (first case) and the rMLMC estimator (third case) possess comparable behavior as a function of the desired precision level $\epsilon$. The second estimator (i.e. practical parallel variant of the standard estimator) is significantly less efficient (i.e. the computational cost grows higher than $1/\epsilon^2$) than the other two. On the other hand, in terms of expected completion time, it turns out that the second estimator is superior to the third estimator (rMLMC), which in turn is superior to the standard estimator (i.e. the first estimator). However, we also study a variation of the rMLMC estimator, which introduces a controlled bias due to the use of a suitable time truncation, such variation matches the expected completion time growth of the second estimator without sacrificing the performance in the growth of the expected computational cost. This feature (i.e. the use of a suitable truncation in the design of rMLMC) highlights additional design aspects of important consideration that are briefly discussed in Section \ref{sub:practical_MLMC}.

Similarly, in the context of MCMC for geometrically ergodic Markov chains, we study the advantages of unbiased MCMC strategies. We conclude that in this setting, due to the geometric ergodicity and practical considerations in terms of various tuning parameters for algorithmic design, unbiased MCMC provide significant advantages to biased estimators. The conclusions, reported in Section \ref{sec:MCMC} follow from our results in earlier sections. 

The rest of this paper is structured as follows: Section \ref{sec:cost_analysis} establishes the mathematical framework and explores the total cost and completion time required for a generic algorithm to achieve $\epsilon^2$-MSE. Section \ref{subsec:unbias_analysis} clarifies the relationship between completion time and the tail behavior for unbiased methods. Section \ref{sec:MLMC} and \ref{sec:MCMC} apply the general results in earlier sections to MLMC and MCMC respectively, as indicated earlier. Section \ref{sec:numerical} provides a numerical example to substantiate our theoretical findings. The paper concludes with a brief discussion in Section \ref{sec:discussion}.

While the primary focus of the paper is to establish a mathematical framework for comparing unbiased versus biased methods, some remarks on practical implementations are also included at the end of Section \ref{sec:MLMC} and \ref{sec:MCMC} to complement our results.

\section{Cost Analysis under different metrics}\label{sec:cost_analysis}

\subsection{Problem Formulation}
In a given probability space $(\Omega, \mathcal F, \mathbb P)$, we are interested in a quantity that can be expressed as a functional of the probability measure, denoted as $\mathcal T(\mathbb P)$. We consider a randomized algorithm, $\mathcal A$, which takes an input $x$ and produces a random variable, $\mathcal A(x)$. \footnote{Throughout the paper, the output of an algorithm or the resulting Monte Carlo estimator will refer to the same concept.} This output serves as an estimate of $\mathcal T(\mathbb P)$. When the input is irrelevant or absent for our discussion, we simply denote the algorithm as $\calA$.  The computational expense of producing a single replication of an estimator (i.e. running the algorithm once) is given by $C(\calA)$. We measure $C(\calA)$ in terms of random seeds needed to run $\calA$. We recognize that our model of computation is somewhat vague, however, our objective is to compare the performance of unbiased estimators with their natural counterparts in practical settings (as discussed in Sections \ref{sec:MLMC} and \ref{sec:MCMC}). In these contexts there will be no ambiguity and our results will yield precise comparisons as long as the computational cost is measured consistently. We allow $C(\mathcal{A})$ to be a random variable, particularly when considering algorithms that terminate at random completion times. Observe that many algorithms, like MCMC, have an iterative structure requiring an initial input, say $x$. In the context of such iterative algorithms, we employ \( \mathcal{A}:=\mathcal{A}(x,n) \) to signify the output following \( n \) iterations and use \( C(\mathcal{A},n) \) to denote the associated computational cost.

We also assume that users have access to a sufficient number of processors. This allows them to generate independent and identically distributed (i.i.d.) random variables by executing the same algorithm independently across different processors. If the algorithm $\mathcal A$ is run on $m$ processors, the expected (total) \textit{computational cost} can be defined as the sum $\mathbb E[C(\mathcal{A}_1) + C(\mathcal{A}_2) + \ldots + C(\mathcal{A}_m)] = m \mathbb E[C(\mathcal{A}_1)]$. Here, $C(\mathcal {A}_i)$ represents the cost incurred on the $i$-th processor. 

We also consider the  worst-case and average-case \textit{completion times}, denoted as $\max_i C(\mathcal{A}_i)$ and $\sum_{i=1}^mC(\mathcal{A}_i)/m$, respectively. Although one might be interested in other statistical metrics, such as a function of the worst or average completion time, our study will not study them in details here for the sake of specificity.

These concepts are illustrated by the following simple examples.

\begin{example}
    Suppose $\mathcal A$ is an algorithm  with deterministic cost, i.e., $C(\mathcal{A})$ is a constant. When users run  $\mathcal A$ on $m$ processors, the computational cost is simply $m C(\mathcal{A})$. Meanwhile, the worst-case and the average-case completion time coincides, both equals $C(\mathcal{A})$. 
\end{example}

\begin{example}
    Consider the following algorithm $\calA$: continue flipping a coin until a head shows up, and output the number of flips. This algorithm has a random cost, where $C(\calA)$ is a geometric distribution with success parameter $1/2$. 
\end{example}

\begin{example}
    Consider any standard Markov chain Monte Carlo algorithm $\mathcal A$ targeting at $\mathbb P$. Though the algorithm is itself random, the computational cost for running $\mathcal A$  $n$ steps can be viewed as deterministic. In contrast, the coupling algorithm described in Section \ref{subsec:coupling} below has a random cost for each run.
\end{example}

Throughout this paper, we will fix an error tolerance level $\epsilon$, and study/compare the completion time and completion cost between certain unbiased algorithm and its biased counterpart when they both achieve a mean-squared error (MSE) at most $\epsilon^2$. Assuming the output of $\mathcal A$ has bias $b_{\mathcal {A}}$ and variance $v_{\mathcal {A}}$, then the MSE of implementing $\mathcal A$ on $m$ processors would have MSE $ = b_{\mathcal {A}}^2 + v_{\mathcal {A}}/m$. For iterative algorithms, the bias $b_\calA(n)$ and variance $v_\calA(n)$ of running $\mathcal A$ depends on the number of iterations $n$. Implementing $\calA$ for $n$ steps independently on $m$ processors have MSE = $b_\calA(n)^2 + v_\calA(n)/m$. Therefore, users need to choose $n$ and $m$ suitably large such that  the  summation is at most $\epsilon^2$. 

\subsection{Analysis of unbiased algorithms}\label{subsec:unbias_analysis}

The following simple result characterizes the computation cost and the completion time for an unbiased algorithm to achieve $\epsilon^2$-MSE.

\begin{theorem}\label{thm:cost-unbiased}
    Suppose $\calA$ is an algorithm satisfying $b_\calA = 0$, $v_\calA < \infty$, and $\bE[C(\calA)]< \infty$. Then the  expected total computation cost for $\cal A$ to achieve $\epsilon^2$-MSE equals $$C_\calA v_\calA \cdot \epsilon^{-2}.$$ Meanwhile, if $m$ processors are available, then
    users can allocate the calculations suitably such that 
    the expected worst-case computational cost is no more than $$\bE[\max_{1\leq i \leq m(\epsilon)} \sum_{j = 1}^{n(\epsilon)} C(\calA_{i,j})]$$ 
    where $\calA_{i,j}$ are independent repititions of $\calA$, $m(\epsilon) := \min\{m, \ceil{v_\calA \epsilon^{-2}}\}, n(\epsilon) := \ceil{v_\calA \epsilon^{-2}/m(\epsilon)}$.
    The average-case computation cost is no more than $n(\epsilon) \cdot \bE[C(\calA)]$.
    In particular, when sufficiently many (more than $\ceil{v_\calA \epsilon^{-2}}$) processors are available, users can allocate the computing resources such that the worst-case completion time is 
   $\bE[\max_{1\leq i \leq \ceil{v_\calA \epsilon^{-2}} } C(\calA_{i})]$, and average completion time is $\bE[C(\calA)]$.    
\end{theorem}

\begin{proof}[Proof of Theorem \ref{thm:cost-unbiased}]
   It is clear that executing $\calA$ $m$ times and averaging the outcomes produces an estimator with an MSE of $v_{\calA}/m$. Consequently, users can select $m = \ceil{v_\calA \epsilon^{-2}}$ to ensure that the derived estimator meets the $\epsilon^2$-MSE criterion. Therefore, the total computational cost equals $C_\calA v_\calA \cdot \epsilon^{-2}$, as announced.

   Additionally, when $m$ processors are accessible, users can instruct each processor to perform an equal number of repetitions (subject to rounding). In this scenario, every processor conducts at most $\ceil{v_\calA \epsilon^{-2}/m(\epsilon)}$ repetitions, creating an upper limit
$$\bE[\max_{1\leq i \leq m(\epsilon)} \sum_{j = 1}^{n(\epsilon)} C(\calA_{i,j})]$$
for the expected worst-case completion time and
$$
n(\epsilon) \bE[C(\calA)],
$$
for the expected average-case completion time.  

In instances where more than $\ceil{v_\calA \epsilon^{-2}}$ processors are available, then $n(\epsilon) = 1$. That is, we can simply assign the initial $\ceil{v_\calA \epsilon^{-2}}$ processors to each execute $\calA$ once. The relevant outcomes for the completion times are directly inferred from the prior general scenario.
\end{proof}
\begin{remark}
    The above results still holds if we replace $v_\calA$ by any upper bound of  $v_\calA$. This may be practically relevant, as mostly $v_\calA$ is not known to users as a-priori, but an upper bound may be available. In the remaining of the paper we refer to the expected worst case completion time simply as the worst case completion time.
\end{remark}

Although the Theorem \ref{thm:cost-unbiased} is mathematically simple, it reveals two important phenomena. Firstly, Theorem \ref{thm:cost-unbiased} suggests an unbiased Monte Carlo algorithm would have a $\Theta(\epsilon^{-2})$ computational cost, provided it  has finite variance and finite expected computation cost. The $\Theta(\epsilon^{-2})$ cost is known as the optimal  complexity for general Monte Carlo methods \citep{heinrich1999monte}. Many existing unbiased algorithms are shown to have finite variance and computation cost under different contexts \citep{glynn2014exact,jacob2020unbiased,rhee2015unbiased,blanchet2015unbiased, goda2022unbiased}.

However, the most interesting aspect of \ref{thm:cost-unbiased} is the portion involving the completion time analysis. In particular, the worst-case completion time of an unbiased algorithm crucially depends on the maxima of $i.i.d.$ random variables distributed according to $C(\cal A)$. Our subsequent results will quantify the completion time by utilizing the tail distribution of $C(\calA)$. Before stating them, we first formally define sub-exponential and sub-Gaussian random variables, which are widely used in the study of concentration phenomena\footnote{Unfortunately, the concept of subexponential distributions is also used in the study of large deviations of random walks with heavy-tailed increments, \cite{EKM2013}. We adopt the concentration literature language which in particular implies that a sub-exponential random variable has light tails.}. 

\begin{definition}[Sub-gaussian]\label{def:sub-gaussian}
    A  random variable \(X\) is said to be \textit{sub-Gaussian} with parameter \(\sigma^2\) if there exists a constant \(\sigma > 0\) such that its moment-generating function satisfies:
\[
\mathbb{E}[e^{\lambda (X- \bE[X])}] \leq \exp\left(\frac{\sigma^2 \lambda^2}{2}\right)
\]
for all \(\lambda \in \mathbb{R}\).
\end{definition}

\begin{definition}[Sub-exponential]\label{def:sub-exponential}
    A  random variable \(X\) is said to be \textit{sub-exponential} with parameters \((\nu, \alpha)\) if there exist constants \(\nu > 0\) and \(\alpha > 0\) such that its tail probabilities satisfy:
\[
\mathbb{P}(X-\bE[X] > t) \leq \exp\left(-\frac{t}{\nu}\right) \quad \text{and} \quad \mathbb{P}(X -\bE[X]< -t) \leq \exp\left(-\frac{t}{\alpha}\right)
\]
for all \(t > 0\).
\end{definition}
With these definitions in hand, we are able to upper bound the completion time.

\begin{theorem}\label{thm: completion-moment}
    With all the notations the same as Theorem \ref{thm:cost-unbiased}. Suppose sufficient many (more than $\ceil{v_\calA \epsilon^{-2}}$) processors are available, the worst-case completion time satisfies  
    \begin{itemize}
        \item If $C(\cal A)$ has a finite expectation, then
        $\bE[\max_{1\leq i \leq \ceil{v_\calA \epsilon^{-2}} } C(\calA_{i})] = o(\epsilon^{-2})$,
        \item If $C(\cal A)$ has a finite $p$-th moment, then
        $\bE[\max_{1\leq i \leq \ceil{v_\calA \epsilon^{-2}} } C(\calA_{i})] = o(\epsilon^{-2/p})$,
        \item If $C(\cal A)$ is sub-exponential with parameter \((\nu, \alpha)\) , then 
         \[\bE[\max_{1\leq i \leq \ceil{v_\calA \epsilon^{-2}} } C(\calA_{i})] = \bE[C(\calA_{1})] + \nu\left(\log \ceil{v_\calA \epsilon^{-2}}  + 1\right) = O(\log(\epsilon^{-1}))\]
         \item If $C(\cal A)$ is sub-Gaussian with parameter $\sigma^2$, then 
         $$\bE[\max_{1\leq i \leq \ceil{v_\calA \epsilon^{-2}} } C(\calA_{i})]  \leq\bE[C(\calA_{1})] +  \sigma \sqrt{2 \log \ceil{v_\calA \epsilon^{-2}} }  + \frac{1}{\sqrt{2\log 2}} = O(\sqrt{\log(\epsilon^{-1})}),$$
    \end{itemize}
   
\end{theorem}

\begin{proof}
We prove the four claims stated in Theorem \ref{thm: completion-moment} one by one. 
\begin{itemize}
    \item When $C(\cal A)$  has a finite expectation, we will show the following two results.
    \begin{enumerate}
        \item The sequence $$\frac{\max\{C(\calA_{1}), \ldots, C(\calA_{n}))\}}{n} \rightarrow 0$$ almost surely, and thus converges to $0$ in probability. 
        \item The sequence $$\frac{\max\{C(\calA_{1}, \ldots, C(\calA_{n}))\}}{n} \rightarrow 0$$ is uniformly integrable (see Section 4.6 of \cite{durrett2019probability} for a definition).          
    \end{enumerate}
    These two result togetherly yields 
    \[
    \bE\left[\frac{\max\{C(\calA_{1}, \ldots, C(\calA_{n}))\}}{n}\right] \rightarrow 0
    \]
    according to Theorem 4.6.3 in \cite{durrett2019probability}. Then taking $n = \ceil{v_\calA \epsilon^{-2}}$ gives the desired result. 
To prove almost sure convergence, one first observes $C(\calA_{n})/n$ converges to $0$ almost surely. This is because for any $\epsilon > 0$, 
\begin{align}
    \sum_{n=1}^\infty \bP[C(\calA_{n})/n> \epsilon] &= \sum_{n=1}^\infty \bP[C(\calA_{n})> n\epsilon]  \\ 
    & = \sum_{n=1}^\infty \bP[C(\calA_{1})> n\epsilon]\\   
    & \leq \int_0^\infty \bP[C(\calA_{1}) > x] \diff x/\epsilon\\
    & = \bE[C(\calA_{1})]/\epsilon < \infty.
\end{align}
By the Borel–Cantelli lemma, almost surely, the sequence of events $\{C(\calA_{n})/n > \epsilon\}$ only occurs finitely many times, which in turn implies $C(\calA_{n})/n$ almost surely converges to $0$.

Meanwhile, according to the following calculus result:
If $a_n/n \rightarrow 0$, then we have $\max\{a_1, \ldots, a_n\} \rightarrow 0$ (can be proven directly using the $\epsilon-N$ method), we conclude that $$\frac{\max\{C(\calA_{1}), \ldots, C(\calA_{n}))\}}{n} \rightarrow 0$$ almost surely. 

To show uniform integrability, we will show the sequence $\{\sum_{i=1}^n C(\calA_{i})/n\}$ is uniformly integrable. Since $\max\{C(\calA_{1}), \ldots, C(\calA_{n})) \leq \sum_{i=1}^n C(\calA_{i})$, the former result directly implies the uniform integrability of $\{\max\{C(\calA_{1}), \ldots, C(\calA_{n}))/n\}$ by definition.

To prove the sequence $\{\sum_{i=1}^n C(\calA_{i})/n\}$ is uniformly integrable, we first observe that the sequence of $\{C(\calA_{n})\}$ is uniformly integrable, which can be derived directly from the definition. Then for any set $E$, we have
\begin{align*}
    \bE\left[\sum_{i=1}^n \frac{\sum_{i=1}^n C(\calA_{i})}{n} I(E)\right] = \frac{\sum_{i=1}^n \bE[C(\calA_i) I(E)]}{n} \leq \sup_{i\leq n} \bE[C(\calA_i) I(E)].
\end{align*}
 Now we conclude that for any $\epsilon > 0$, there is some $\delta > 0$ such that whenever $\bP[E] < \delta$, $\sup_{n}\bE\left[\sum_{i=1}^n \frac{\sum_{i=1}^n C(\calA_{i})}{n} I(E)\right] \leq \epsilon$, which follows from our previous inequality and the uniformly integrability of $\{C(\calA_{n})\}$. Therefore $\{\max\{C(\calA_{1}), \ldots, C(\calA_{n})\}/n\}$ is uniformly integrable, as desired. 
\item Applying the previous result to $\{C(\calA_{n})^p\}$, we have 
\[
\bE\left[\max_{1\leq i \leq n}(C(\calA_{i}))^p \right] = o(n).
\]
Since $\max_{1\leq i \leq n}(C(\calA_{i}))^p = (\max_{1\leq i \leq n}C(\calA_{i}))^p$,  taking $1/p$-th power on both sides yields
\[
\bE\left[\max_{1\leq i \leq n}C(\calA_{i}) \right] = o(n^{1/p}).
\]
Again, taking  $n = \ceil{v_\calA \epsilon^{-2}}$ gives the desired result. 
\item We can bound the tail probability of $\max_{1\leq i \leq n}C(\calA_{i})$ by:
\begin{align*}
    \bP[\max_{1\leq i \leq n}C(\calA_{i}) > t + \bE[C(\calA_{1})]] &\leq n\bP[ C(\calA_{1}) > t + \bE[C(\calA_{1})]]\\
    & \leq n \exp\left(-\frac{t}{\nu}\right).
\end{align*}
Setting $T =\bE[C(\calA_{1})] +  \nu \log n $, we have:
\begin{align*}
    \bE[\max_{1\leq i \leq n}C(\calA_{i}) ] &= \bE[\max_{1\leq i \leq n}C(\calA_{i}) I(\max_{1\leq i \leq n}C(\calA_{i}) \leq T)]  + \bE[\max_{1\leq i \leq n}C(\calA_{i}) I(\max_{1\leq i \leq n}C(\calA_{i}) > T)] \\
    &\leq T + \int_{\nu\log n}^\infty n \exp\left(-\frac{t}{\nu}\right) \diff t\\
    & = T + \mu n \exp(-\log n ) \\
    & = \bE[C(\calA_{1})] + \nu(\log n + 1).
\end{align*}

\item Similarly, taking $\lambda = t/\sigma^2$,we have  the  following tail bound of $\max_{1\leq i \leq n}C(\calA_{i})$: 
\begin{align*}
    \bP[\max_{1\leq i \leq n}C(\calA_{i}) > t + \bE[C(\calA_{1})]] &\leq n\bP[ C(\calA_{1}) > t + \bE[C(\calA_{1})]]\\
    & = n \bP[ \exp\left(\lambda(C(\calA_{1}) - \bE[C(\calA_{1})])\right) > \exp(\lambda t)]  \\
    & \leq n \exp(-\lambda t) \exp(\frac{\lambda^2 \sigma^2}{2})\\
    & = n \exp(-\frac{t^2}{2\sigma^2})
\end{align*}

For any $n\geq 2$, setting $T =\bE[C(\calA_{1})] +  \sigma \sqrt{2 \log n} $, we have:
\begin{align*}
    \bE[\max_{1\leq i \leq n}C(\calA_{i}) ] &= \bE[\max_{1\leq i \leq n}C(\calA_{i}) I(\max_{1\leq i \leq n}C(\calA_{i}) \leq T)]  + \bE[\max_{1\leq i \leq n}C(\calA_{i}) I(\max_{1\leq i \leq n}C(\calA_{i}) > T)] \\
    &\leq T + \int_{\sigma \sqrt{2 \log n}}^\infty n \exp(-\frac{t^2}{2\sigma^2}) \diff t\\
    & \leq  T + \sigma n \frac{\exp(-\log n)}{\sigma\sqrt{2 \log n} } \\
    & \leq\bE[C(\calA_{1})] +  \sigma \sqrt{2 \log n}  + \frac{1}{\sqrt{2\log 2}},
\end{align*}
 where the second to last inequality follows from the standard Gaussian tail bound 
    \[
    \int_{x}^\infty \exp(-y^2/2) \diff x \leq x^{-1} \exp(-x^2/2)
    \]
    for every $x > 0$ in Theorem 1.2.6 of \cite{durrett2019probability}.
 \end{itemize}
\end{proof}
To summarize, Theorem \ref{thm: completion-moment} explores the relationship between the worst-case completion time and the tail behavior of each implementation. Specifically, a lighter tail in the implementation of $\mathcal{A}$ results in quicker completion times, and vice versa. For instance, if the tail probability $\bP[C(\mathcal A) > x]$ is only slightly lighter than $1/x$ (which corresponds to the case of finite expectation), then the most we can infer is that the completion time is of a lower order of magnitude than the total cost $\Theta(1/\epsilon^2)$. More specific predictions are not possible. If the tail probability is marginally less than $1/x^p$, then the completion time scales as $o(\text{Completion Time}^{1/p})$. Likewise, with an exponential tail, the completion time grows logarithmically relative to the total cost. If the tail resembles a Gaussian distribution (super-exponential), the completion time scales as the square root of the logarithm of the total cost. We will apply this result on both unbiased MCMC and MLMC algorithms in later sections.

\subsubsection{A limit theorem for the completion time}
Theorem \ref{thm: completion-moment} gives upper bounds for the completion time given the information of either moments or upper bounds on the tail probability. In cases where more information on the distribution of $C(\calA)$ is available, we are able to establish limit theorem for the completion time. We start with the well-known Fisher-Tippett-Gnedenko Theorem.

\begin{theorem}[Fisher-Tippett-Gnedenko Theorem]\label{thm:FTG Theorem}
Let $X_1, \ldots, X_n$ be $n$ i.i.d. random variables. Suppose there  exist two sequences of real numbers $a_n > 0$ and $b_n \in \mathbb R$ such that the sequence of random variables $$\frac{\max\{X_1, \ldots, X_n\} - b_n}{a_n}$$ converges (in distribution) to a non-degenerate distribution $G$. Then $G$ must be of the form:
\begin{align*}
    G(x) = \exp\left\{-\left( 1 + \xi \frac{x-\mu}{\phi}\right)_+^{-1/\xi}\right\},
\end{align*}
where $y_+ = \max\{y,0\}$, $\phi > 0$, and $\mu,\xi$ are arbitrary real numbers. The case $\xi = 0$ should be understood as the limit $\xi \rightarrow 0$, corresponding to the distribution
\begin{align*}
    G(x) = \exp\left\{\exp\left(-\left(\frac{x-\mu}{\phi}\right)_+\right) \right\}.
\end{align*}
The limiting distribution $G$ takes one of the three fundamental types, namely the Fréchet, Gumbel, and Weibull types. Let $F$ be the cumulative distribution function (CDF) of $X_1$ and  
$x^\star := \sup \{x: \mathbb P(X_1 \leq x) \leq 1\}$ be the population maximum. Then the limiting distribution $G$ of the normalized maximum will be:
\begin{itemize}
	\item A Fréchet distribution ($\xi > 0$) if and only if:
	\[x^\star = \infty \quad \text{and} \quad \lim_{t\rightarrow \infty}\frac{1 - F(ut)}{1 - F(t)} = u^{-1/\xi} \quad \text{for all } u > 0.\]

	\item A Gumbel distribution ($\xi = 0$) if and only if: 
	\[
	\lim_{t\rightarrow x^\star} \frac{1 - F(t + uf(t))}{1 - F(t)} = e^{-u} \quad \text{for all } u > 0, \text{with } f(t) = \frac{\int_{t}^{x\star} 1 - F(s) \diff s}{1 - F(t)}.
	\]

	\item A Weibull distribution ($\xi < 0$) if and only if:
	\[
	x^\star < \infty \quad \text{and} \quad  \lim_{t\rightarrow 0^+}\frac{1 - F(x^\star - ut)}{1 - F(x^\star - t)} = u^{1/\lvert \xi \rvert}  \quad \text{for all } u > 0.
	\]

\end{itemize}
\end{theorem}

The limiting distribution together with the normalizing sequence $\{a_n, b_n\}$ characterizes the order of the quantile distribution. The next corollary is immediate:
\begin{corollary}
	Let $X_1, \ldots, X_n$ be $n$ i.i.d. random variables. Suppose there  exist two sequences of real numbers $a_n > 0$ and $b_n \in \mathbb R$ such that the sequence of random variables $$\frac{\max\{X_1, \ldots, X_n\} - b_n}{a_n}$$ converges (in distribution) to a non-degenerate distribution $G$. Let $m(n,q)$ be the $q$-th quantile of the maxima  $\max\{X_1, \ldots, X_n\}$ and $G^{-1}(q)$ be the $q$-th quantile of $G$. Then 
	\[
	\lim_{n\rightarrow \infty} \frac{m(n,q) - b_n}{a_n} = G(q).
	\]
\end{corollary}
 In particular, this implies the following result on the completion time.  
\begin{prop}\label{prop:maxima-limiting}
	With all the notations the same as before, for any fixed $q\in(0,1)$, we have
	\begin{itemize}
		\item If the tail of $C(\calA)$ is regular varying, i.e., $\bP[C(\calA)> x] = C(1 + o(1)) x^{-\gamma}$ for $C,\gamma > 0$ as $x\rightarrow \infty$, then the $q$-th quantile of $\max_{1\leq i \leq \ceil{v_\calA \epsilon^{-2}}}C(\calA_{i}) $ is at the order of $\epsilon^{-2/\gamma}$.
		\item If $$\bP[C(\calA) > x] = (1 + o(1)) e^{-\alpha x}$$ as $x\rightarrow \infty$, then the $q$-th quantile of  $\max_{1\leq i \leq \ceil{v_\calA \epsilon^{-2}}}C(\calA_{i})$ is$$  (1+o(1))\log(\ceil{v_\calA \epsilon^{-2}})/\alpha. $$
		\item If $C(\cal A)$ is normally distributed, then the $q$-th quantile of $\max_{1\leq i \leq \ceil{v_\calA \epsilon^{-2}}}C(\calA_{i}) $ is 
              $$
             \sqrt{v(\calA)} \left(\sqrt{2\log(\ceil{v_\calA \epsilon^{-2}})} + \bE[C(\calA)]\right) + o(1)
              $$
	\end{itemize}
\end{prop}

\begin{proof}
    \begin{itemize}
        \item If the tail of $C(\calA)$ is regular varying, solving the equation 
        $$\bP[C(\calA)> x] = C/n,$$
          we have $x = (1 + o(1)) n^{1/\gamma}$. Therefore, setting $x_n = n^{1/\gamma}$
          \begin{align*}
              \bP[\max\{C(\calA_1), \ldots, C(\calA_n)\} > \lambda x_n] &= 1 -  (1 -  \bP[C(\calA) >\lambda x_n])^n \\
              & = 1 - \left(1 - \frac{C(1+o(1)) \lambda^{-\gamma}} {n}\right)^n\\
              & = 1 - \exp{(-C\lambda^{-\gamma})}.
          \end{align*}
          Therefore Theorem \ref{thm:FTG Theorem} holds with the choice of $b_n = 0, a_n = x_n$, $\xi = 1/\gamma$ and the limiting distribution is the Fréchet distribution. Thus the $q$-th quantile is at the order of $a_{\ceil{v_\calA \epsilon^{-2}}} = \epsilon^{-2/\gamma}$.
          \item If $\bP[C(\calA) > x] = (1 + o(1)) e^{-\alpha x}$, we can similarly solve the $\bP[C(\calA)> x] = 1/n,$ equation and get $ x= (1+o(1))(\log n)/\alpha$. Therefore, setting $b_n  = \log (n)/\alpha$
        \begin{align*}
              \bP[\max\{C(\calA_1), \ldots, C(\calA_n)\} > \lambda + b_n] &= 1 -  (1 -  \bP[C(\calA) >\lambda + b_n])^n \\
              & = 1 - \left(1 - \frac{(1+o(1)) \exp{-(\alpha\lambda)}} {n}\right)^n\\
              & = 1 - \exp{(-\exp{(-(\alpha\lambda))})}.
          \end{align*}
          Therefore the $q$-th quantile of $\ceil{v_\calA \epsilon^{-2}}$ is $(1+o(1))\log(\ceil{v_\calA \epsilon^{-2}})/\alpha$.

          \item For i.i.d. standard normal random variables $X_1, \ldots X_n,$ it is well-known that if we choose $c_n$ satisfying $c_n/\sqrt{2\log n}  \rightarrow 1$, then
          \begin{align*}
              \lim_{n\rightarrow\infty} \bP\left[\frac{\max\{X_1, \ldots, X_n\} - c_n}{1/c_n} \leq u\right] = G_0(u),
          \end{align*}
          where $G_0(u) = \exp\{-\exp\{-u\}\}$ is the cumulative distribution function of the Gumbel distribution. Therefore, applying the above formula on the normalized version of $C(\calA)$, i.e., $(C(\calA) - \bE[C(\calA)])/\sqrt{v(\calA)}$, and set $n = \ceil{v_\calA \epsilon^{-2}}$ imply the desired result. 
    \end{itemize}
\end{proof}

When comparing the findings of Theorem \ref{thm:FTG Theorem} and Proposition \ref{prop:maxima-limiting} against those of Theorem \ref{thm: completion-moment}, it becomes evident that Theorem \ref{thm:FTG Theorem} offers greater accuracy. However, it also necessitates a more detailed understanding of the distribution of \( C(\calA) \). For instance, Theorem \ref{thm: completion-moment} establishes that if the tail probability of \( C(\cal A) \) can be constrained by a function that decays exponentially, then the completion time is bounded by \( O(\log(\epsilon^{-1})) \). On the other hand, Theorem \ref{thm:FTG Theorem} and Proposition \ref{prop:maxima-limiting} indicate that if the tail probability \( \bP[C(\calA)>x] \) is asymptotically equal to \( \exp(-\alpha x) \), the completion time will also be asymptotically equal to \( \log(\ceil{v_\calA \epsilon^{-2}})/\alpha \). Another distinguishing feature is the scope of applicability: Theorem \ref{thm: completion-moment} is non-asymptotic, making it  valid for any \( n \) or \( \epsilon \). In contrast, Theorem \ref{thm:FTG Theorem} is an asymptotic limiting theorem, making it most appropriate for scenarios where \( n \) is exceptionally large or \( \epsilon \) is exceedingly small.

Both theorems have their respective merits and can be particularly useful depending on the context in which they are applied. In some cases, the distribution of \( C(\calA) \) may be  predefined prior to executing the algorithm. As a result, we can employ Theorem \ref{thm:FTG Theorem} in conjunction with Proposition \ref{prop:maxima-limiting} to derive a highly accurate estimate of the completion time. 
In other cases, such as working with unbiased MCMC, the algorithm typically revolves around designing a coupling, and completion occurs once this coupling takes place. While users can frequently obtain qualitative insights into the distribution of the coupling time, such as its exponential decay characteristics, before running the algorithm, obtaining a detailed distribution of \( C(\calA) \) is usually not feasible. In scenarios like this, Theorem \ref{thm: completion-moment} becomes more informative for assessing the algorithm's performance.

\subsection{Analysis of biased algorithms}
Recall that \( b_{\mathcal{A}} \) represents the bias of algorithm \( \calA \). Given that running \( \calA \) independently across \( m \) processors results in an MSE of \( b_{\mathcal{A}}^2 + v_{\mathcal{A}}/m \geq b_{\mathcal{A}}^2 \), achieving an MSE of at most \( \epsilon^2 \) necessitates initially configuring \( \calA \) so that its bias is strictly less than \( \epsilon^2 \). To clarify this dependency, we introduce \( \calA_{\epsilon/2} \) and \( C(\calA_{\epsilon/2}) \) to signify the low-biased algorithm with a maximum bias of \( \epsilon/2 \) and its associated computational cost, respectively. The constant \( 1/2 \) in this context is flexible and can be selected from a range between \( 0 \) and \( 1 \). Although this selection could impact the constant factors in our cost analysis, it will not affect the order of magnitude under consideration.

The next theorem is immediate from the bias-variance decomposition just stated.

\begin{theorem}\label{thm:cost-biased}
    Suppose $\calA_{\epsilon/2}$ is an algorithm satisfying $b_{\calA_{\epsilon/2}} <\epsilon/2$, $v_{\calA_{\epsilon/2}} < \infty$, and $\bE[C(\calA_{\epsilon/2})]< \infty$. Then the computation cost for $\calA_{\epsilon/2}$ to achieve $\epsilon^2$-MSE is $$O( \bE[C(\calA_{\epsilon/2})] v_{\calA_{\epsilon/2}}\cdot \epsilon^{-2}).$$ Meanwhile, if $m$ processors are available, then the minimal worst-case computational cost is no more than $$\bE[\max_{1\leq i \leq m(\epsilon)} \sum_{j = 1}^{n(\epsilon)} C(\calA_{i,j,\epsilon/2})]$$ 
    where $\calA_{i,j,\epsilon/2}$ are independent repititions of $\calA_{\epsilon/2}$, $m(\epsilon) := \min\{m, \ceil{0.75v_{\calA_{\epsilon/2}} \epsilon^{-2}}\}$, and \\$n(\epsilon) := \ceil{0.75v_{\calA_{\epsilon/2}}\epsilon^{-2}/m(\epsilon)}$.
    The average-case computation cost is no more than $n(\epsilon) \cdot \bE[C(\calA_{\epsilon/2})]$.
    In particular, if the cost $C(\calA_{\epsilon/2})$ is deterministic and
    sufficient many (more than $\ceil{v_\calA \epsilon^{-2}}$) processors are available, the worst-case completion time  and average completion time are both $C(\calA_{\epsilon/2})$.    
\end{theorem}

Upon comparing Theorem \ref{thm:cost-biased} with Theorems \ref{thm:cost-unbiased} and \ref{thm: completion-moment}, it becomes evident that the completion time for the low-biased algorithm \( \calA_{\epsilon/2} \) is primarily governed by its intrinsic cost \( C(\calA_{\epsilon/2}) \). In contrast, the completion time of the unbiased algorithm \( \calA \) is chiefly determined by its tail behavior. As we will demonstrate in the examples presented in the subsequent two sections, when the tail of \( C(\calA) \) is sufficiently light (such as an exponential tail), the unbiased algorithm often completes faster than its biased equivalent, and vice versa.

\subsection{ A remark on other metrics}
When running algorithm \( \calA \) across \( m \) processors, it is natural to employ an \( m \)-dimensional vector, denoted as \( \mathsf{cost} := (C(\calA_1), C(\calA_2), \ldots, C(\calA_m)) \in \mathbb R^{\geq 0} \), to encapsulate the computational cost for each processor. In this context, the total cost corresponds to the \( L^1 \) norm of \( \mathsf{cost} \), while the worst-case cost is represented by its \( L^\infty \) norm. Practically speaking, users might opt for \( \lVert \mathsf{cost} \rVert_p \) (the general \( L^p \) norm), or \( f(\lVert \mathsf{cost} \rVert_p) \) (a utility function based on the \( L^p \) norm), or even a hybrid utility function that combines \( L^1 \) and \( L^\infty \) norms. The theoretical scrutiny of these alternative norms should align closely with the analysis we've conducted above. For the sake of focus, we won't delve into these options in depth.

\section{Applications on Multilevel Monte Carlo}\label{sec:MLMC}
\subsection{Non-randomized Multilevel Monte Carlo: Optimal total cost and  completion time analysis}
Multilevel Monte Carlo (MLMC) methods aim to estimate a specific attribute of a target distribution, such as the expectation of a function applied to the solution of a Stochastic Differential Equation (SDE) \citep{giles2008multilevel}, a function of an expectation \citep{blanchet2015unbiased}, solutions of stochastic optimization problems, quantiles, and steady-state expectations; see, for instance, \cite{blanchet2015unbiasedtaylor,blanchet2015unbiased,blanchet2019unbiased}. Typically, users have access to a hierarchy of approximations for the target distribution, where each successive approximation features reduced variance but increased computational cost. For example, when the objective is to determine $\bE[f(X_T)]$, with $X_T$ representing the solution of a stochastic differential equation (SDE) at time $T$, users can employ finer discretization to achieve greater accuracy, albeit at a higher cost. The primary objective of MLMC is to estimate this value with a high degree of accuracy and minimal computational cost by optimally combining these approximations. The following theorem offers theoretical guarantees regarding the computational complexity of multilevel Monte Carlo methods.

\begin{theorem}[Theorem 1 in \cite{giles2015multilevel}]\label{thm:MLMC}
Let $P$ denote a random variable whose expectation we want to estimate.
Suppose for each $l\geq 0$, we have an algorithm $\calA^l$ that outputs a random variable $\Delta_l$  with variance $V_l$ and computational cost $C_l$. Define $s_l := \sum_{k=0}^l \bE[\Delta_k]$  and assume the following holds for some $\alpha, \beta, \gamma, c_1, c_2, c_3$ such that $\alpha \geq \frac{1}{2}\max\{\gamma, \beta\}$ and:
\begin{itemize}
    \item $\lvert s_l - \bE[P]\rvert \leq c_1 2^{-\alpha l}$;
    \item $V_l \leq c_2 2^{-\beta l}$;
    \item $C_l \leq c_3 2^{\gamma l}$.
\end{itemize}
Then for any fixed $\epsilon < 1/e$, there exists an estimator $Y$ satisfying $\bE[(Y- \bE[P])^2] < \epsilon^2$ with computational cost:

    $$
\begin{cases}
\calO(\epsilon^{-2}), ~~~\qquad\qquad \gamma < \beta\\
\calO(\epsilon^{-2}(\log\epsilon)^2), \qquad \gamma = \beta\\
\calO(\epsilon^{-2 - (\gamma - \beta)/\alpha}), \qquad \gamma > \beta
\end{cases}
$$
\end{theorem}
 Giles'  algorithm starts by simulating the target distribution at each level of approximation using a fixed number of samples $N_l$. For each level $l$, the algorithm generates $N_l$ independent copies of the random variable $\Delta_l^{(1)},\ldots, \Delta_l^{(N_l)}$. Then the final estimator is given by $$\hat P := \sum_{l = 0}^L \frac{1}{N_l} \sum_{i=1}^{N_l}\Delta_l^{(i)}.$$ The above theoretical guarantee is attained by choosing $N_l$ optimally. 

We will concentrate specifically on the case where \( \gamma < \beta \), because in this scenario, the optimal cost of \( \mathcal{O}(\epsilon^{-2}) \) is achievable.  With $\epsilon$ fixed, the parameters of Giles' Multilevel Monte Carlo algorithm can be chosen as $$L =\log_2(2c_1/\epsilon)/\alpha =  \calO(\log(1/\epsilon)) $$ and 
 $$
 N_l = 0.75^{-1} \epsilon^{-2}\sqrt{\frac{V_l}{C_l}} \sum_{k=0}^L \sqrt{C_k V_k}.
 $$
This ensures the  estimator has at most $\epsilon^2$-MSE. The total cost would be 

\begin{align*}
    \sum_{l=0}^L N_l C_l &= 0.75^{-1} \epsilon^{-2} (\sum_{k=0}^L \sqrt{C_kV_k})^2 \\
    & \leq  0.75^{-1} \epsilon^{-2} (\sum_{k=0}^\infty \sqrt{C_kV_k})^2  = \calO(\epsilon^{-2}).
\end{align*}

While Giles's algorithm, represented by \( \mathcal{A}_G \), attains the optimal \( O(\epsilon^{-2}) \) cost, it is designed so that \( b_{\mathcal{A}_G} \leq \frac{\epsilon}{2} \) and \( v_{\mathcal{A}_G} \leq 0.75 \epsilon^2 \) to meet the \( \epsilon^2 \)-MSE criterion. This means that running \( \mathcal{A}_G \) on multiple processors to further reduce variance is unnecessary. The algorithm can be run on a single processor, achieving both the  cost and completion time of \( O(\epsilon^{-2}) \).

If the aim is to minimize execution time and there are enough processors at hand, the simplest approach to modify Giles's algorithm (so that every processor receives the same amount of work) involves configuring a low-bias algorithm \( \mathcal{A}_{\epsilon/2} \) such that 

\[
L = \frac{\log_2(2c_1/\epsilon)}{\alpha} = \mathcal{O}(\log(1/\epsilon))
\]

and setting \( N_l = 1 \) for all \( 0 \leq l \leq L \) - compare with the choice in Theorem \ref{thm:MLMC}. This results in a deterministic cost \( C(\mathcal{A}_{\epsilon/2})= O(\epsilon^{-\gamma/\alpha}) \) and \( v_{\mathcal{A}} = O(1) \). By utilizing \( \mathcal{A}_{\epsilon/2} \) on \( \lceil 0.75^{-1}v_{\mathcal{A}_{\epsilon/2}}\epsilon^{-2}/m(\epsilon) \rceil \) processors and averaging the outcomes, the final estimator still satisfies the \( \epsilon^2 \)-MSE condition. 

The completion time for \( \mathcal{A}_{\epsilon/2} \) in this case would be \( O(\epsilon^{-\gamma/\alpha}) \), which is generally  favorable compared to Giles's standard estimator which features \( O(\epsilon^{-2}) \) completion time under the assumption \( \alpha \geq 0.5 \max\{\gamma,\beta\} \); see Theorem \ref{thm:MLMC}. However, the total computational cost for this approach would be \( O(\epsilon^{-(2+\gamma/\alpha)}) \), which will always be less efficient than that of Giles's algorithm.

Consequently, Giles's original algorithm offers an optimal total cost in a non-parallel environment. On the other hand, the low-bias algorithm \( \mathcal{A}_{\epsilon/2} \), which is an easy-to-implement adaptation of Giles' algorithm for parallelization, provides a more favorable completion time but at the expense of a suboptimal total computational cost.

\subsection{Randomized Multilevel Monte Carlo}

Giles' algorithm has been extended in several important ways, including Multi-index Monte Carlo \citep{haji2016multi} and the randomized MLMC (rMLMC) algorithm \citep{rhee2015unbiased}. Notably, the rMLMC algorithm is unbiased and also attains the optimal total cost. The algorithm, denoted by $\calA_{\rMLMC}$ is described below: 

\begin{algorithm}[htbp]
\caption{$\calA_{\rMLMC}$: Unbiased Randomized Multilevel Monte-Carlo estimator}\label{alg:rMLMC}

\begin{algorithmic}

\State \textbf{Input:} A probability mass function $p(n)$ supporting on all non-negative integers \\
\begin{enumerate}
	\item Sample $N$ from the discrete distribution 
	\item Call Algorithm $\calA^N$ defined in Theorem \ref{thm:MLMC}, and generate a random variable $\Delta_N$
\end{enumerate}
\State\textbf{Return: } $W: =\Delta_N/p_N$.
\end{algorithmic}
\end{algorithm}

The theoretical properties of Algorithm \ref{alg:rMLMC} are established in \cite{rhee2015unbiased}. We now summarize them below. 

\begin{theorem}\label{thm:rMLMC}
Given the same notations as previously discussed and under the assumption that \( \gamma < \beta \), selecting \( p(n) \propto 2^{-(\beta+\gamma)n/2} \) ensures that the output \( W \) of Algorithm \ref{alg:rMLMC} is unbiased, has finite variance, and incurs a finite computational cost in expectation. Consequently, by running \( \mathcal{A}_{\text{rMLMC}} \) for \( \lceil v_{\mathcal{A}_{\rMLMC}}\epsilon^{-2} \rceil \) repetitions and taking the average, we obtain an estimator within expected total cost $\calO(\epsilon^{-2} )$ that satisfies the \( \epsilon^2 \)-MSE criterion at most.
\end{theorem}

Theorem \ref{thm:rMLMC} demonstrates that $\calA_{\rMLMC}$ achieves the optimal cost of $\calO(\epsilon^{-2})$, aligning it with Giles's original algorithm and making it more efficient than the naive parallel version $\calA_{\epsilon/2}$ discussed in the prior section. Now we turn to study its completion time. 

\begin{prop}\label{prop: single completion}
    Maintaining the same conditions outlined in Theorem \ref{thm:rMLMC}, the expected worst-case completion time required to execute $\calA_{\rMLMC}$ across all processors to obtain an estimator with a mean squared error (MSE) of $\epsilon^2$ satisfies, for any $c>0$,
    \[
     \bE[\max_{1\leq i \leq \ceil{v_\calA \epsilon^{-2}} } C(\calA_{\rMLMC, i})] = o(\epsilon^{-4\gamma/(\beta+\gamma) + c})
    \]
\end{prop}

\begin{proof}
    Each call of $\calA_{\rMLMC}$ costs $2^{\gamma N}$ where $\bP[N = n] \propto 2^{-(\beta+\gamma)n/2}$. Therefore one can directly show $C(\calA_{\rMLMC})$ has a finite $2^{((\beta+\gamma)/2\gamma) - \delta}$-th moment for any $\delta \in (0, (\beta+\gamma)/2\gamma)$. According to Theorem \ref{thm: completion-moment}, the completion time over all processors to get an estimator with $\epsilon^2$-MSE satisfies
    \[
     \bE[\max_{1\leq i \leq \ceil{v_\calA \epsilon^{-2}} } C(\calA_{\rMLMC,i})] = o(\epsilon^{-4\gamma/(\beta+\gamma) + c})
    \]
    for any $c > 0$. 
\end{proof}

Note that the expected completion time $o(\epsilon^{-4\gamma/(\beta+\gamma) + c})$ is strictly better than $O(\epsilon^{-2})$, as $\beta > \gamma$ implies $4\gamma/(\beta+\gamma) < 2$. On the other hand, since $\alpha \geq \max\{\beta, \gamma\}/2$, we also have $4\gamma/(\beta+\gamma) \geq \gamma/\alpha$ which indicates the completion time of $\calA_{\rMLMC}$ is typically worse than $\calA_{\epsilon/2}$.

To briefly summarize, we are evaluating three algorithms: Giles's original MLMC denoted as $\calA_{G}$, a naive parallel version of Giles's algorithm labeled as $\calA_{\epsilon/2}$, and the randomized MLMC represented as $\calA_{\rMLMC}$. When it comes to overall computational cost, both $\calA_{G}$ and $\calA_{\rMLMC}$ attain the optimal $O(\epsilon^{-2})$ cost, outperforming $\calA_{\epsilon/2}$, which has the highest cost of \( O(\epsilon^{-(2+\gamma/\alpha)}) \). In terms of worst-case completion time, $\calA_{G}$ also has a rate of $O(\epsilon^{-2})$, as it was not initially conceived to be parallelized. On the other hand, $\calA_{\rMLMC}$ offers a more favorable completion time of $o(\epsilon^{-4\gamma/(\beta+\gamma) + c})$. Lastly, $\calA_{\epsilon/2}$ finishes the quickest with a completion time of \( O(\epsilon^{-(\gamma/\alpha)}) \). From this comparison, it seems the randomized MLMC algorithm offers a good balance between the total cost and the completion time. 

Certainly, there might be additional metrics and variations that users might want to consider. For example, when evaluating based on average-case completion time, $\calA_{\rMLMC}$ surpasses all other methods with an $O(1)$ average completion time. However, worst-case completion time could be more relevant in practical scenarios, as a task is only completed when all processors have finished their work.

One might naturally wonder if it's possible to modify $\calA_{\rMLMC}$ to improve the worst-case completion time, such as by truncating the discrete sample size \( N \) when it becomes excessively large. As outlined in Algorithm \ref{alg:rMLMC-bias}, such modifications are indeed feasible.

\begin{algorithm}[htbp]
\caption{$\tilde{\calA}_{\rMLMC}$: Truncated, low-biased Randomized Multilevel Monte-Carlo estimator}\label{alg:rMLMC-bias}

\begin{algorithmic}

\State \textbf{Input:} A probability mass function $p(n)\propto \{0,1, \ldots, L\}$, where $L =\log_2(2c_1/\epsilon)/\alpha$ \\
\begin{enumerate}
	\item Sample $N$ from the discrete distribution 
	\item Call Algorithm $\calA^N$ defined in Theorem \ref{thm:MLMC}, and generate a random variable $\Delta_N$
\end{enumerate}
\State\textbf{Return: } $W: =\Delta_N/p_N$.
\end{algorithmic}
\end{algorithm}
The theoretical guarantee is established in the Proposition below. The proof is largely the same as before and is therefore omitted. 
\begin{prop}\label{prop: completion-truncated}
   Given the same notations as previously discussed and under the assumption that \( \gamma < \beta \), selecting \( p(n) \propto 2^{-(\beta+\gamma)n/2} \) ensures that the output \( W \) of Algorithm \ref{alg:rMLMC-bias} has bias less than $\epsilon/2$, has finite variance, and incurs a finite computational cost in expectation. Consequently, by running \( \tilde{\mathcal{A}}_{\text{rMLMC}} \) for \( 0.75^{-1}\lceil v_{\tilde{\mathcal{A}}_{\rMLMC}}\epsilon^{-2} \rceil \) repetitions and taking the average, we obtain an estimator within expected total cost $\calO(\epsilon^{-2} )$ that satisfies the \( \epsilon^2 \)-MSE criterion at most. Moreoever,  the worst-case completion time required to execute $\tilde{\calA}_{\rMLMC}$ (Algorithm \ref{alg:rMLMC-bias}) across all processors satisfies:
    \[
     \bE[\max_{1\leq i \leq 0.75^{-1}\lceil v_{\tilde{\mathcal{A}}_{\rMLMC}}\epsilon^{-2} \rceil } C(\tilde{\calA}_{\rMLMC, i})] = O(\epsilon^{-\gamma/\alpha})
    \]
\end{prop}
From a theoretical standpoint, $\tilde{\mathcal{A}}_{\text{rMLMC}}$ appears to excel in both key areas, enjoying the optimal total computational cost of $O(\epsilon^{-2})$ as well as the most favorable worst-case completion time of $O(\epsilon^{-\gamma/\alpha})$ when compared to other methods. The only attribute it lacks theoretically is unbiasedness. In the following section, we will delve further into the practical aspects of implementing these algorithms.

\subsection{Remarks on practical implementation and further extensions}\label{sub:practical_MLMC}
While the main focus of this paper is to offer a theoretical analysis comparing these methods in a parallel computing environment, we also include some observations regarding their real-world implementation.

\begin{itemize}
    \item Required knowledge of the parameters: For effective deployment of Giles's algorithm, denoted as $\calA_G$, it's essential to understand all parameters—$\alpha, \beta, \gamma, c_1, c_2, c_3$. This proves to be quite challenging in practical scenarios. As noted in the review paper \cite{giles2015multilevel}, it is common that $c_1$ and $c_2$ are not known a priori but must be estimated empirically, based on weak error and multilevel correction variance. The current analysis does not take into account the cost of estimating these parameters, and how the errors may effect the corresponding computational cost. On the other hand, the implementation of the $\calA_{\rMLMC}$ framework demands only that the relationship $\beta > \gamma$ be satisfied, a condition achievable in numerous contexts as supported by \cite{rhee2015unbiased, blanchet2015unbiased, wang2022unbiased}. As for the two other low-bias algorithms—$\calA_{\epsilon/2}$ and $\tilde{\calA}_{\rMLMC}$—which serve as variations of $\calA_G$ and $\calA_{\rMLMC}$ respectively, knowledge of $\alpha$ and $c_1$ is required but not of $c_2$ and $c_3$. Hence, even though $\tilde{\calA}_{\rMLMC}$ appears to offer the most robust theoretical guarantees on paper, $\calA_{\rMLMC}$ may still be more straightforward to implement in a variety of cases. 
    \item Dependence on $\epsilon$:  Among the four algorithms discussed, $\calA_{\rMLMC}$ stands out as the only one that does not necessitate a predetermined $\epsilon$ value for its implementation. Therefore, another appealing feature of $\calA_{\rMLMC}$ is its 'adaptive' stopping capability, meaning that users are not required to specify the number of repetitions to run in advance. Once an accuracy level $\epsilon$ is set, one can implement $\calA_{\rMLMC}$ in parallel while concurrently computing the Monte Carlo confidence interval. The simulation process can be halted when the length of this confidence interval falls below the predefined $\epsilon$. The legitimacy of this adaptive stopping criterion is provided in the work \cite{glynn1992asymptotic}.

    \item Generalizability: The simplicity and unbiased nature of $\calA_{\rMLMC}$ make it inherently more adaptable to complex problems when used as a subroutine. For instance, in the context of finite-horizon, discrete-time optimal stopping issues, the paper \cite{zhou2021unbiased} introduces an estimator that achieves a computational complexity of $O(1/\epsilon^2)$ by recursively employing the $\calA_{\rMLMC}$ estimator. This approach is further extended in \cite{syed2023optimal}.

    \item Numerical Stability: In our implementation experience, the randomized methods $\calA_{\rMLMC}$ and $\tilde{\calA}_{\rMLMC}$, which typically employ the ratio-based estimator $\Delta_N/p_N$, appear to be numerically less stable compared to the summation-based estimators found in $\calA_{G}$ and $\calA_{\epsilon/2}$. This instability becomes more noticeable when the value of $N$ is exceptionally large, although such instances are uncommon. Specifically, $\calA_{\rMLMC}$ exhibits even less stability than $\tilde{\calA}_{\rMLMC}$ because the latter constrains the value of $N$ to ensure a lower bound for the denominator. Recent investigations, such as \cite{giles2023efficient, haji2023nested}, explore the potential of combining $\calA_{G}$ and $\calA_{\rMLMC}$ for estimating nested expectations. This hybrid approach could potentially harness the strengths of both algorithms.

    \item Topics not covered: Our current exploration of the unbiased algorithm $\calA_{\rMLMC}$ is centered on the condition where $\beta > \gamma$. In scenarios where $\beta \leq \gamma$, it becomes challenging to simultaneously achieve unbiasedness, finite expectation, and finite computational cost. Nonetheless, similar truncation methods to those employed in the creation of $\tilde{\calA}_{\rMLMC}$ can also be applied to formulate a low-bias randomized MLMC algorithm. This adapted algorithm can achieve the same overall computational cost and a reduced completion time compared to $\calA_G$, provided enough processors are accessible. Note that additional required information, in comparison to $\calA_{\rMLMC}$, includes knowledge of $\alpha$ and \(c_1\). Lastly, the parallelization approach we have examined in this paper is fundamentally of the `embarrassingly parallel' variety, where users run the algorithm independently without intercommunication. It is, however, feasible—though more nuanced—to develop a parallelization strategy for $\calA_G$ that allocates resources at the subroutine level, specifically for $\calA^1, \ldots, \calA^L$, as outlined in Theorem \ref{thm:MLMC}. Though this remains an intriguing avenue for future research, we do not explore this more complex form of parallelization in the current study for the sake of concentration.
\end{itemize}

\section{Applications on Markov Chain Monte Carlo}\label{sec:MCMC}

\subsection{Standard MCMC: cost and completion time analysis}
Monte Carlo Markov Chain (MCMC) methods are a class of algorithms used for sampling from complex probability distributions in order to estimate expectations of certain variables of interest. These methods combine the principles of Monte Carlo simulation, which involves random sampling to approximate numerical results, with Markov Chains, which represent a sequence of events in which the probability of each event depends solely on the state attained in the previous event. MCMC algorithms, such as Metropolis-Hastings or Gibbs sampling, iteratively generate a chain of samples in such a way that the chain eventually converges to the target distribution. Once this convergence is achieved, the samples can be used to approximate various statistical properties, including expected values, of the distribution. This approach is particularly useful in Bayesian statistics and machine learning, where the distribution in question may be complex, high-dimensional, or not easily described by analytical formulas. For a comprehensive overview, we direct readers to the works of \cite{diaconis2009markov,brooks2011handbook}.

In practical applications, one frequently encounters the challenge of estimating the expectation represented as:
\begin{align}\label{eqn:estimation}
	\bE_\pi[f] = \int f(x)\pi(dx)
\end{align}
where it's difficult to sample directly from the underlying distribution \( \pi \). Often, a Markov chain \( \Phi = (\Phi_1, \Phi_2, \ldots, \Phi_n) \) can be constructed to target \( \pi \) as its stationary distribution. Under mild  conditions, it is well-established \cite{geyer1992practical, jones2001honest} that the standard MCMC estimator
\[
	H_{\text{MCMC}}(n) = \frac{\sum_{i=1}^n f(\Phi_i)}{n}
\]
converges to \( \mathbb{E}_\pi(f) \) as \( n \) approaches infinity. The Markov chain Central Limit Theorem applies under slightly stronger, yet still manageable, conditions, see  \cite{roberts2004general,jones2004markov} for details. In practice, users often discard an initial set of samples, a process known as the `burn-in' technique \cite{geyer2011introduction}.  Since employing a constant-length `burn-in' phase will not alter the scale or findings of any upcoming theoretical examinations, we will continue assuming that no `burn-in' is utilized for the rest of this section \footnote{While `burn-in' is commonly used, some experts debate its necessity. For more details, see \url{http://users.stat.umn.edu/~geyer/mcmc/burn.html}.}.

To continue discussing the performance of standard MCMC estimators, it is useful to recall some definitions in Markov chain theory. We say a $\pi$-invariant, $\phi$-irreducible and aperiodic Markov transition kernel $P$ is \textit{geometrically ergodic} if there exists a constant $\rho < 1$, and a function $M(x)$ such that for $\pi$-a.s. $x$ and every $n$,
\begin{equation}
  \left\| P^n(x, \cdot) - \pi(\cdot) \right\|_\TV \leq M(x) \rho^n.
\end{equation}
Here $\lVert \cdot \rVert_\TV$ stands for the total variation (TV) distance between two probability measures. The technical definitions for irreducibility and aperiodicity can be found in Chapter 5 of \cite{meyn2012markov}.

The next Theorem collects sufficient conditions for the Markov chain Strong Law of Large Numbers (SLLN) and Central Limit Theorem (CLT).

\begin{theorem}\label{thm:MCLLNCLT}
    Suppose $\Phi$ is a $\pi$-invariant, $\phi$-irreducible and aperiodic  Markov chain with transition kernel $P$ on a state space with countably generated $\sigma$-algebra, then 
    \begin{itemize}
        \item (Markov chain SLLN, Theorem 17.0.1 of \cite{meyn2012markov} or Fact 5 in \cite{roberts2004general}) For any function $f$ with $\bE_\pi[\lvert f \rvert] < \infty$, 
        \begin{align*}
            \lim_{n\rightarrow\infty} \frac 1n \sum_{i=1}^n f(\Phi_i) =  \bE_\pi[f] \qquad \text{almost surely.}
        \end{align*}
        \item (Markov chain CLT, Theorem 24 of \cite{roberts2004general})In addition, if $\Phi$ is geometrically ergodic and $\bE[\lvert f \rvert^{2+\delta}] < \infty$ for some $\delta > 0$, then
        \begin{align*}
            \lim_{n\rightarrow\infty} \sqrt{n}\left(\frac 1n \sum_{i=1}^n f(\Phi_i)  - \bE_\pi[f]\right) \overset{d}{\to} \mathbb{N}(0, \sigma_f^2)
        \end{align*}
        for some $\sigma_f > 0$. 
    \end{itemize}
\end{theorem}

If we represent the procedure of `executing the MCMC algorithm for \( n \) steps and then producing \( H_{\text{MCMC}}(n) \)' as \( \mathcal{A}_{\text{MCMC}}(n) \), and consider both a single iteration of the MCMC algorithm and a single arithmetic operation as having unit cost, then the preceding theorem immediately implies the following conclusion regarding the total computational expense for standard MCMC estimators to achieve an \( \epsilon^2 \)-MSE.

\begin{prop}\label{prop:totalcost_mcmc}
Under the CLT assumptions in Theorem \ref{thm:MCLLNCLT}, if a sufficiently large constant \( C > 0 \) exists such that when \( n =  \ceil{C\epsilon^{-2}} \), the MSE of \( \mathcal{A}_{\text{MCMC}}(n) \) is no greater than \( \epsilon^2 \). The total computational cost of executing this algorithm is \( \mathcal{O}(\epsilon^{-2}) \).
\end{prop}

Consequently, assuming the chain is geometrically mixing and the function of interest has moments exceeding the second under the stationary distribution, standard MCMC algorithms also achieve the optimal \( \mathcal{O}(\epsilon^{-2}) \) computational cost. However, parallelizing these inherently sequential MCMC algorithms presents its own challenges. A straightforward approach to parallelization involves independently running \( \mathcal{A}_{\text{MCMC}}(n) \) on each processor, choosing \( n \) large enough so that the resulting bias is less than \( \epsilon/2 \). As established in the literature (for example, see page 24 of \cite{geyer2011introduction}), the bias \( b_{\mathcal{A}_{\text{MCMC}}}(n) \) is \( \mathcal{O}(1/n) \). The subsequent proposition outlines the time to completion for such a straightforward parallelization approach. The proof directly follows from  stated results in Markov chain theory and is thus skipped here.

\begin{prop}\label{prop:completion_mcmc}
Under the CLT assumptions in Theorem \ref{thm:MCLLNCLT}, suppose that large enough constants \( C_1, C_2 > 0 \) exist such that the algorithm \\\( \mathcal{A}_{\epsilon/2} := \mathcal{A}_{\text{MCMC}}(C_1\epsilon^{-1}) \) meets the conditions \( b_{\mathcal{A}_{\epsilon/2}} \leq \epsilon/2 \) and \( v_{\mathcal{A}_{\epsilon/2}} \leq C_2\epsilon \). Hence, by running \( \mathcal{A}_{\epsilon/2} \) independently on \( 0.75^{-1} C_2\epsilon^{-1} \) processors and averaging the outputs, an estimator is produced with an \( \epsilon^2 \)-MSE within a completion time of \( C_1\epsilon^{-1} \) and total cost $C_1 C_2 \epsilon^{-2}$.
\end{prop}

Consequently, this straightforward method of parallelization maintains the overall computational cost at \( \mathcal{O}(\epsilon^{-2}) \), while reducing the time to completion to \( \mathcal{O}(\epsilon^{-1}) \) provided that an adequate number of processors are accessible. While the estimator $\calA_{\epsilon/2}$ appears to offer superior theoretical performance compared to the standard single processor estimator, its practical implementation can be challenging. This is primarily because it necessitates prior knowledge of the constants $C_1$ and $C_2$, which are often difficult to obtain in practice. We will discuss the practical aspects of these algorithms later in this section. For example, another standard estimator $\calA_{\epsilon/2}$ that can be considered in the context of aperiodic (geometrically ergodic) chains involves running the estimator for $O(\log (1/\epsilon))$. This estimator while competitive in theory has the problem that it needs apriori bounds on rates of convergence to be implemented. These and other practical considerations are discussed at the end of this section.  

\subsection{Unbiased estimator by coupling}\label{subsec:coupling}
Although the Strong Law of Large Numbers (SLLN) and the Central Limit Theorem (CLT) for Markov Chains affirm the consistency of the standard MCMC estimator, this estimator typically exhibit an $\calO(1/n)$ bias for \( \mathbb{E}_\pi[f] \) when \( n \) is fixed, unless the chain starts from its stationary distribution.  Now we briefly review the idea of the unbiased  MCMC algorithms, which is originated from \cite{glynn2014exact} and then used in \cite{jacob2020unbiased} for general-purpose MCMC algorithms.   The idea is based on the following debiasing formula used in  \cite{glynn2014exact}:
\begin{equation}\label{eqn:debiasing}
	\bE_\pi[f] = \bE[f(\Phi_k)] + \sum_{i = k+1}^\infty \bE[f(\Phi_i) - f(\Phi_{i-1})],
\end{equation}
where $k$ is an arbitrary  fixed integer. A coupling consists of a pair of Markov chains, denoted as $(Y,Z) = (Y_t, Z_t)_{t=1}^\infty$, which exist on a combined state space and are governed by a joint transition kernel represented by $\bar P$. Both chains in the pair follow the same original transition kernel $P$, which also guides the evolution of the original chain $\Phi$. To initialize the coupled chain, $Y_0$ is drawn from  some simple distribution $\pi_0$ and $Y_1$ is drawn from $P(Y_0, \cdot)$, while $Z_0$ is independently sampled from $\pi_0$. The evolution of the pair then proceeds according to $(Y_t, Z_{t-1}) \sim \barP((Y_{t-1}, Z_{t-2}), \cdot)$. This allow us to rewrite formula \eqref{eqn:debiasing} as
\begin{align*}
    \bE_\pi[f]= \lim_{n\rightarrow \infty} \bE[f(Y_n)] &=  \bE[f(Y_k)] + \sum_{n=k+1}^\infty (\bE[f(Y_{n})] - \bE[f(Y_{n-1})])\\
	&= \bE[f(Y_k)] + \sum_{n=k+1}^\infty\bE[f(Z_{n}) - f(Y_{n-1})] 
\end{align*}

The most important element in designing a coupling is to find a `faithful' coupling \citep{rosenthal1997faithful} which can be easily implemented without  extra information of the underlying chain. In a faithful coupling, there exists a random yet finite time $\tau$ at which $Y_{\tau} = Z_{\tau-1}$, and once $Y_t$ and $Z_{t-1}$ meet, they remain identical thereafter. This concept allows us to reformulate the previous equation as:

\begin{align*}
    \bE_\pi[f] = \bE[f(Y_k)] + \sum_{n=k+1}^{\tau-1}\bE[f(Z_{n}) - f(Y_{n-1})] 
\end{align*}

For any value of $k$, the estimator $H_k(Y,Z) := f(Y_k) + \sum_{i=k+1}^{\tau-1}(f(Y_i) - f(Z_{i-1}))$ serves as an unbiased estimator for $\bE_\pi[f]$. The proof confirming its unbiased nature is presented in \cite{jacob2020unbiased}. Both theoretical and empirical examinations of these approaches are covered in \cite{o2021metropolis, wang2021maximal, papp2022new}. We will again assume $k = 0$ as any fixed $k$ will not alter the magnitude analysis discussed below. 

From now on,  denote the unbiased MCMC algorithm which outputs $H_0(Y,Z)$ by $\calA_{\uMCMC}$. 
The theoretical properties of $\calA_{\uMCMC}$ as well as their completion time requires 
a bit more information of the underlying chain.  We say a $\pi$-invariant, $\phi$-irreducible and aperiodic Markov transition kernel $P$ satisfies a geometric drift condition if there exists a measurable function $V: \Omega \rightarrow [1,\infty)$, $\lambda\in (0,1)$, and a measurable set $\calS$ such that for all $x\in\Omega$:
\begin{align}\label{eqn:drift}
	\int P(x,\diff y) V(y) \leq \lambda V(x) + b\Indc(x\in \calS).
\end{align}
Moreover, the set $\calS$ is called a small set if there exists a positive integer $m_0$, $\epsilon > 0$, and a probability measure $\nu$ on such that for every $x\in \calS$:
\begin{align}\label{eqn: small}
	P^{m_0}(x,\cdot) \geq  \epsilon \mu(\cdot).
\end{align}
Typically, these results are applied with $m_0=1$. The next result  describes the tail behavior of the distribution of $\tau$:

\begin{prop}[Proposition 4 in \cite{jacob2020unbiased}]\label{prop: moment, formal}
	Suppose the Markov transition kernel  satisfies a geometric drift condition with a small set $\calS$ of the form $\calS =\{x: V(x)\leq L\}$ for $\lambda + b/(1+L) < 1$. Suppose there exists $\tilde \epsilon \in (0,1)$ such that
	\[
	\inf_{(x,y)\in \calS \times \calS} \barP((x,y), \calD) \geq \tilde\epsilon,
	\]
	where $\calD:= \{(x,x): x\in \Omega\}$ is the diagonal of $\Omega\times \Omega$. 	Then 
	there exist constants $C'$ and $\kappa \in (0,1)$ such that
 \[
\bP[\tau > n] \leq C' \pi_0(V)\kappa^n. 
 \]
\end{prop}

The next result summarizes the theoretical properties of $\calA_{\uMCMC}$:
\begin{prop}[Proposition 1 in \cite{jacob2020unbiased}]\label{prop:unbiasedMCMC}
   Assuming the conditions of Proposition \ref{prop: moment, formal} hold, and further assuming that $\bE[f(\Phi_n)]$ converges to $\bE_\pi[f]$ as $n \rightarrow \infty$ and $\bE_\pi[|f|^{2+\delta}] < \infty$, we can conclude the following about $\calA_{\uMCMC}$: it is an unbiased estimator, its variance $v_{\calA_{\uMCMC}}$ is of the order $\calO(1)$, and its cost  $C(\calA_{\uMCMC})$ is sub-exponential.
\end{prop}

Therefore we are able to apply Theorem \ref{thm:cost-unbiased} and \ref{thm: completion-moment} to conclude the following: 
\begin{prop}\label{prop:cost_completion_uMCMC}
    Given the same notations as previously discussed and under the assumptions in Proposition \ref{prop:unbiasedMCMC}, running \( \mathcal{A}_{\uMCMC} \) for \( \lceil v_{\mathcal{A}_{\rMLMC}}\epsilon^{-2} \rceil \) repetitions and taking the average, we obtain an estimator within expected total cost $\calO(\epsilon^{-2} )$ and completion time 
    \[\bE[\max_{1\leq i \leq \lceil v_{\mathcal{A}_{\rMLMC}}\epsilon^{-2} \rceil} \calA_{\uMCMC, i}] = \calO(\log(1/\epsilon))    
    \]
    that satisfies the \( \epsilon^2 \)-MSE criterion.
\end{prop}
Hence, the $\calA_\rMLMC$ algorithm not only attains the same  total cost but also significantly outperforms $\calA_{\MCMC}$ and its direct parallelized variant $\calA_{\epsilon/2}$ in terms of completion time, achieving an $\calO(\log(1/\epsilon))$ completion time.

Finally we comment on the essential assumptions made in Proposition \ref{prop: moment, formal}. First, it is crucial to note that the outcomes of Propositions \ref{prop:unbiasedMCMC} and \ref{prop:cost_completion_uMCMC} remain valid provided the conditions in Proposition \ref{prop: moment, formal} are met at a `qualitative level.' Specifically, the actual implementation of the unbiased MCMC algorithm does not require quantitative knowledge of any of the constants like $\lambda, b, L, C', \kappa$ mentioned in Proposition \ref{prop: moment, formal}. Secondly, these assumptions have a direct relationship with the geometric ergodicity of the MCMC algorithms. Therefore they are known to hold in various scenarios, including random-walk MCMC with Gaussian proposals and targets that have super-exponentially light densities. A wealth of literature supports the validity of these assumptions, such as \cite{mengersen1996rates} for 1-dimensional random-walk Metropolis algorithms, \cite{roberts1996geometric} for multi-dimensional versions, and \cite{roberts1996exponential} for Langevin-type algorithms.

\subsection{Remarks on practical implementation and further extensions}
To complement the preceding theoretical discussion, we offer additional remarks here on practical implementations and potential extensions.

\begin{itemize}
    \item Required knowledge of the parameters: To approximate $\bE_\pi[f]$ with an $\epsilon^2$-MSE using either the conventional MCMC algorithm $\calA_{\MCMC}$ or its direct parallel version $\calA_{\epsilon/2}$, users must have non-trivial knowledge of the chain and may even require an auxiliary algorithm for  uncertainty quantification. In the case of the standard MCMC estimator, this entails estimating $\sigma_f^2$ or creating a running Markov chain confidence interval based on a single chain. Both of these tasks are complex and necessitate special attention, as discussed in \cite{glynn1991estimating, geyer1992practical, flegal2010batch}. For the deployment of $\calA_{\epsilon/2}$, one needs precise quantitative insights—sometimes also referred to as `honest bounds' \citep{jones2001honest}—regarding the rate at which bias diminishes. Achieving this is generally very challenging and akin to accurately bounding the mixing time, particularly in practical applications. Conversely, for methods like the unbiased randomized MLMC, generating confidence intervals and establishing stopping criteria with asymptotic guarantees is comparatively straightforward, as outlined in \cite{glynn1992asymptotic, jacob2020unbiased}.

    \item Difficulty of algorithm design: Creating an effective coupling algorithm for $\calA_{\uMCMC}$ tends to be much more complex than simply executing the standard MCMC method $\calA_{\MCMC}$. Various designs have been put forth for distinct variants, including those cited in \cite{jacob2020unbiased, papp2022new} for Metropolis-Hastings algorithms and Gibbs samplers, \cite{heng2019unbiased} for Hamiltonian Monte Carlo,  \cite{middleton2020unbiased} for pseudo-marginal MCMC, and \cite{ruzayqat2022unbiased,ruzayqat2023unbiased} for continuous-time diffusion processes. This area continues to be a subject of ongoing research. Despite these efforts, there is still a need for the creation of readily usable algorithms.

    \item Topics not covered: Our examination of both biased and unbiased algorithms is presently limited to scenarios where the underlying chain exhibits geometric ergodicity. Nevertheless, these methods can still be applied even when the chain has slower mixing properties, albeit with potential variations in computational expense and completion time. Some insights into the polynomial tail for the coupling time $\tau$ can be found in \cite{middleton2020unbiased}. When it comes to standard MCMC algorithms, we have focused on the no-burn-in or fixed-time burn-in scenarios. If the chain is known to be geometrically ergodic, one could at least theoretically create an $\epsilon$-dependent burn-in approach by eliminating the initial $C\log(1/\epsilon)$ samples and then outputting the subsequent sample. This would yield an estimator with $O(\epsilon)$ bias and complete in $O(\log(1/\epsilon))$ time. Running this algorithm in parallel would result in a completion time comparable to that of $\calA_{\uMCMC}$, albeit with a slightly higher total cost. However, the practicality of this approach is questionable, given that it again requires good estimates on challenging quantities such as the spectral gap.  
\end{itemize}

\section{Numerical analysis: a case study}\label{sec:numerical}
We perform numerical evaluations on the following example using both unbiased and biased approaches. The example is intentionally simple and comes with all the necessary ground truth information, facilitating the verification of our theoretical insights.

Our objective is to estimate the mean of a univariate standard Gaussian using a random-walk Metropolis--Hastings (RWM) algorithm. The proposal distribution is \( Q(z, \cdot) \sim \mathcal{N}(z, \sigma^2) \), and the initial state is \( x_0 = \mathcal{N}(1, 1) \). We adopt a proposal variance of \( \sigma^2 = 2.38^2 \), as suggested in previous studies like \cite{gelman1997weak, christensen2005scaling}. As our unbiased algorithm, we utilize the coupling-based approach detailed in \cite{jacob2020unbiased}. 

Specifically, we run two coupled chains \( \{(Y_t, Z_{t+1})\} \) with identical transition kernels. The initial states are \( Y_0 \sim \mathcal{N}(1, 1) \) and \( Z_0 \sim \mathcal{N}(1, 1) \), and  \( Z_1 \sim \mathcal{N}(Z_0, \sigma^2) \). We employ maximal-reflection coupling for the proposal distributions and use common random numbers for the accept-reject step, as outlined in \cite{jacob2020unbiased, wang2021maximal}. The algorithm concludes when the two chains converge, at a random but almost surely finite time \( \tau \). According to \cite{jacob2020unbiased}, \( Z_0 + \sum_{k=1}^{\tau} Z_{k} - Y_{k-1} \) serves as an unbiased estimator.

To simulate parallel computation, we repeat our coupling algorithm $\calA$ for \( M \) times,  as if it were running on one of \( M \) processors. We then measure the algorithm's performance for each \( M \) ranging from 1 to \( 10^5 \). The completion time \( T(M) \) is recorded as the maximum runtime across all \( M \) repetitions, and the squared error \( \textsf{err}^2(M) \) is calculated as the squared error of the mean estimator over these \( M \) repetitions.

The performance metrics for the unbiased algorithm are summarized in Figure \ref{fig:unbiased_plot}. As indicated in prior research \cite{mengersen1996rates, jacob2020unbiased}, the time it takes for the two chains to couple follows a geometric tail distribution. Consequently, the completion time across \( M \) repetitions should increase no more than logarithmically with \( M \). This behavior is vividly illustrated in the top-left panel of Figure \ref{fig:unbiased_plot}.For instance, the completion time required for 25,000 repetitions is approximately 25, whereas it only slightly exceeds 28 for 100,000 repetitions. Additionally, the top-right panel of the figure clearly demonstrates that the error diminishes exponentially as the completion time increases. (It's worth noting that the vertical axis for squared error is plotted on a logarithmic scale.) This pattern suggests that the completion time required to achieve an \( \epsilon \)-level error should be on the order of \( O(\log(1/\epsilon)) \). Lastly, the two panels at the bottom of the figure depict the decline in error as the number of repetitions \( M \) increases. Theoretically, the expected squared error of the average over \( M \) repetitions should be \( \frac{\text{Var}(\mathcal{A})}{M} \), implying that the logarithm of the squared error should be linearly related to \( \log(M) \) with a slope of -1. This theoretical expectation is empirically confirmed in the bottom-right panel, where the fitted line for \( \log(\textsf{err}^2(M)) \) against \( \log(M) \) has a slope of approximately -1.02.

\begin{figure}[htbp]
    \centering
    \includegraphics[width=\textwidth]{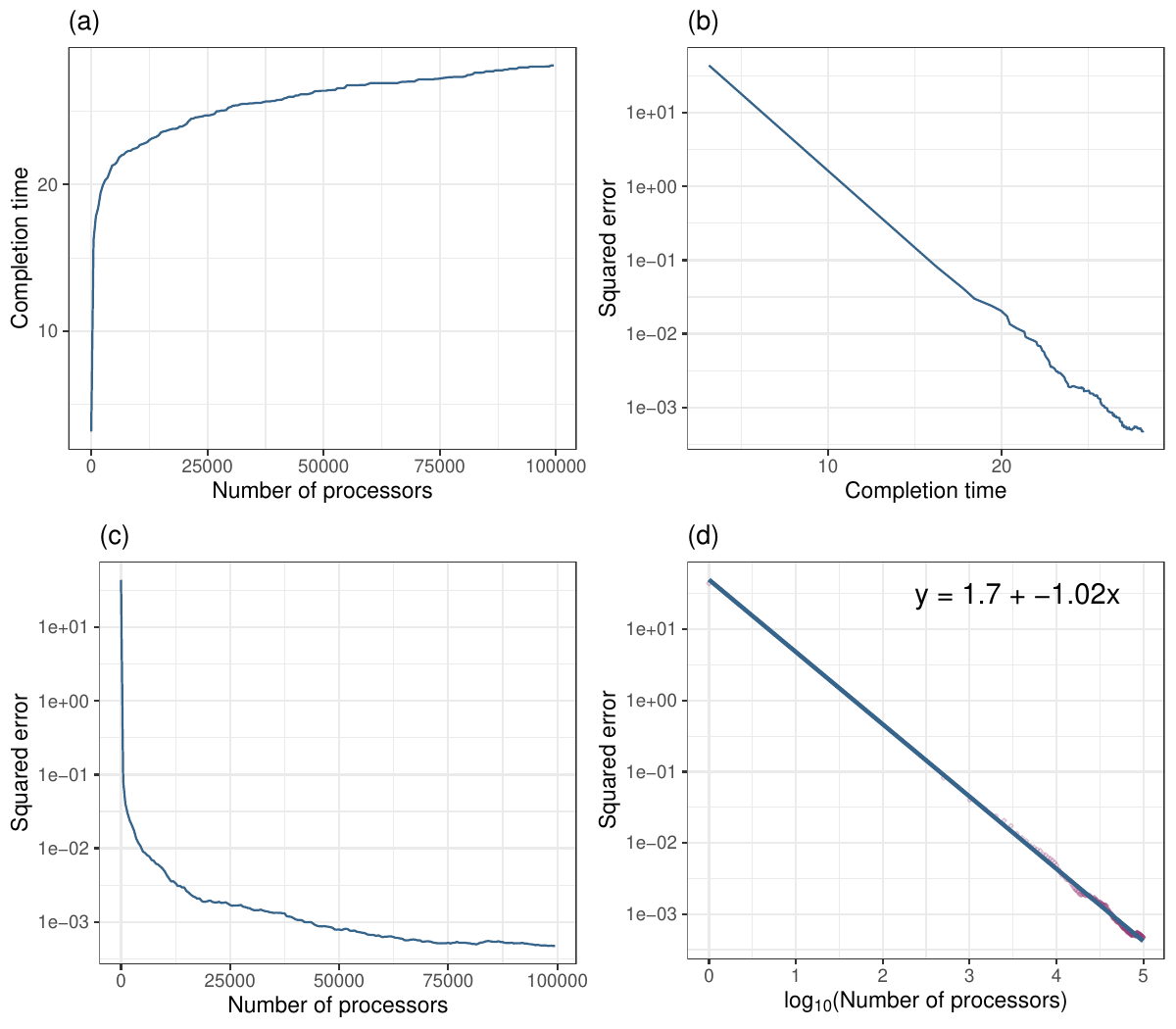}
    \caption{Performance evaluation of the implemented algorithm across $10^5$ processors. Each plot is based on the average results of 100 independent repetitions of the same experiment. (a) The relationship between the number of processors and completion time, demonstrating how parallelization scales. (b) A plot of completion time against the logarithm of squared error. (c) The relationship between the number of processors and the logarithm of squared error. (d) A log-log plot of the number of processors against squared error, with a fitted regression line. }
    \label{fig:unbiased_plot}
\end{figure}

Furthermore, we assess the performance by contrasting the unbiased MCMC method with the standard (biased) MCMC approaches. The unbiased MCMC implementation has been previously detailed. For the standard MCMC methods, we set a specific number of iterations \( N \) (ranging from 1 to 100) and execute the RWM algorithm \( N \) times for \( M \) repetitions (ranging from 1 to \( 10^5 \)), mimicking parallel execution on \( M \) processors. Each combination of \( (N, M) \) yields a specific completion time of \( N \) and a squared error \( \widetilde{\textsf{err}}^2(M, N) \), which is computed as the squared error for the average result over these \( M \) runs. Theoretically, as \( M \) approaches infinity, \( \widetilde{\textsf{err}}^2(M, N) \) should converge to the squared bias inherent to the MCMC algorithm after \( N \) iterations.
\begin{figure}[htbp!]
    \centering
    \includegraphics[width=\textwidth]{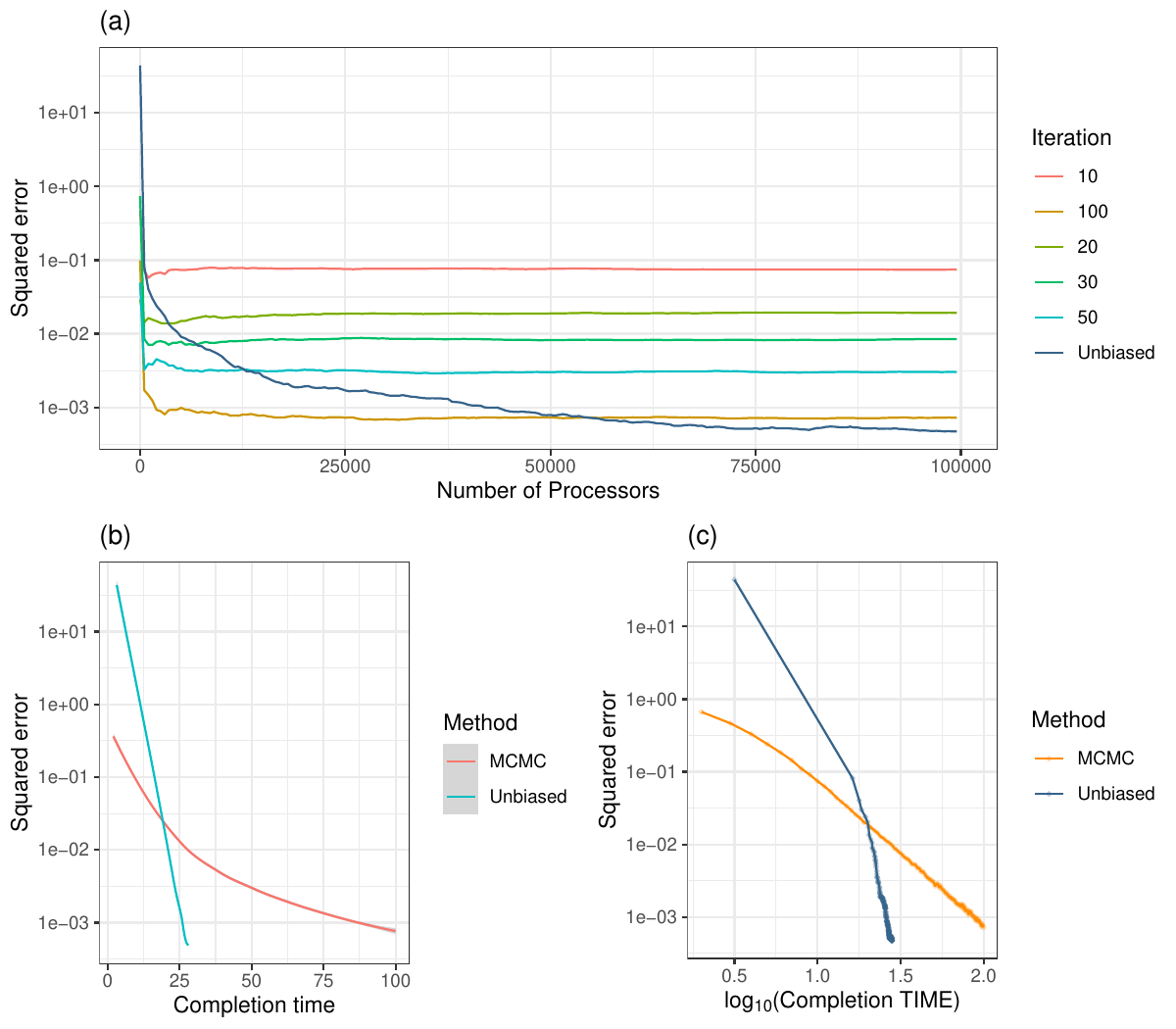}
    \caption{ Comparative Analysis of Unbiased and Biased MCMC Methods. Each plot is based on the average results of 100 independent repetitions of the same experiment.(a) The plot shows how squared error diminishes as the number of processors increases for various methods. The steelblue curve represents the unbiased MCMC method, which has a maximum completion time of \(28.71\). Curves in other colors represent standard MCMC methods conducted for \{10, 20, 30, 50, 100\} iterations. (b) This plot displays the relationship between completion time and squared error for both standard MCMC (in red) and unbiased MCMC (in blue). (c) This plot displays the relationship between the logarithm of completion time and the logarithm of squared error for standard MCMC (in orange) and unbiased MCMC (in steelblue).}
    \label{fig:unbiased_mcmc_comparison}
\end{figure}
The results of the comparative study are presented in Figure \ref{fig:unbiased_mcmc_comparison}. The top panel illustrates the decay of \( \textsf{err}^2(M, N) \) as \( M \) increases for fixed \( N \) values of \( \{10, 20, 30, 50, 100\} \). We have also included the performance of our unbiased method (depicted in steelblue), which reaches a maximum completion time of \( 28.71 \) when employing \( 100,000 \) processors. It is evident that standard MCMC methods are more effective than the unbiased method when the number of processors is limited. This is because a single standard MCMC estimator typically has lower variance, a point also noted in \cite{jacob2020unbiased}. 

However, the advantages of using an unbiased method become increasingly significant as the number of processors increases. Two key observations are worth mentioning here. First, due to its unbiased nature, the error in the unbiased method continues to decline as more processors are utilized. On the other hand, the error in standard MCMC for any fixed number of iterations plateaus after increasing the number of processors, constrained by its inherent bias. Second, when a large number of processors (specifically \( 100,000 \)) are available, the unbiased method, with a maximum completion time of \( 28.71 \), outperforms the standard MCMC method with a completion time of \( 100 \). This further empirically substantiates the benefits of adopting an unbiased approach in a highly parallel computing environment.

To dive deeper into the relationship between completion time and error, we plot the logarithm of the squared error against both the completion time and the logarithm of the completion time for both methods, specifically when \( M = 100,000 \). These plots are displayed in the bottom panels of Figure \ref{fig:unbiased_mcmc_comparison}. The efficacy of the unbiased method in reducing completion time is readily apparent. Furthermore, in Figure \ref{fig:unbiased_mcmc_comparison} (c), we observe that the logarithm of the squared error exhibits near-linear decay for the standard MCMC method and super-linear decay for the unbiased MCMC method. This confirms that the intrinsic bias of the standard MCMC method after \( N \) iterations is on the order of \( 1/N \). In practical terms, to achieve a specific accuracy level \( \epsilon \), the standard MCMC method requires completion time on the order of \( 1/\epsilon \) even when a large number of processors is available. In contrast, the unbiased method necessitates only \( O(\log(1/\epsilon)) \) completion time.

\section{Discussion}\label{sec:discussion}
The insights derived from the comparison between biased and unbiased Monte Carlo methods in parallel computation contexts have clarified instances where unbiased methods can offer a significant advantage and where their advantage may be minimal or nonexistent. Moving forward, unbiased techniques have the potential to be enhanced and developed further through the employment of variance reduction strategies like those discussed in \cite{vihola2018unbiased}, with a focus on diminishing the overall computational expense. Moreover, our present exploration can extend to a wider array of metrics capable of accurately reflecting the practical utility of these methodologies in varied real-world contexts. This will aid practitioners in making informed decisions on the application of these methods. Finally, a promising direction would be the development of hybrid methodologies (see, e.g., \cite{haji2023nested}) that can  leverage the strengths of both unbiased and biased Monte Carlo methods, seeking to strike a balance between computational efficiency, resource utilization, and accuracy. 
	\newpage
\bibliographystyle{chicago}
\bibliography{ref}

\end{document}